\documentclass[11pt]{article}
\usepackage{amsmath,amssymb,amsthm}
\usepackage{graphicx,hyperref,cite}
\usepackage{microtype}
\usepackage{aliascnt}
\usepackage[margin=1in]{geometry}

\newtheorem{theorem}{Theorem}

\newaliascnt{lemma}{theorem} \newtheorem{lemma}[lemma]{Lemma}
\aliascntresetthe{lemma} 

\newaliascnt{corollary}{theorem}

\aliascntresetthe{corollary}

\newaliascnt{definition}{theorem}
\newtheorem{definition}[definition]{Definition}
\aliascntresetthe{definition}

\newaliascnt{observation}{theorem}
\newtheorem{observation}[observation]{Observation}
\aliascntresetthe{observation}

\title{Bipartite and Series-Parallel Graphs\\ Without Planar Lombardi Drawings}
\author{David Eppstein\thanks{Department of Computer Science, University of California, Irvine. Research supported in part by NSF grants CCF-1618301 and CCF-1616248.}}

\begin{document}
\maketitle

\begin{abstract}
We find a family of planar bipartite graphs all of whose Lombardi drawings (drawings with circular arcs for edges, meeting at equal angles at the vertices) are nonplanar. We also find families of embedded series-parallel graphs and apex-trees (graphs formed by adding one vertex to a tree) for which there is no planar Lombardi drawing consistent with the given embedding.
\end{abstract}

\section{Introduction}

\emph{Lombardi drawing} is a style of graph drawing using curved edges. In this style, each edge must be drawn as a circular arc, and consecutive edges around each vertex must meet at equal angles. Many classes of graphs are known to have such drawings, including regular bipartite graphs and all $2$-degenerate graphs (graphs that can be reduced to the empty graph by repeatedly removing vertices of degree at most two)~\cite{DunEppGoo-JGAA-12}. This drawing style can significantly reduce the area usage of tree drawings~\cite{DunEppGoo-DCG-13}, and display many of the symmetries of more general graphs~\cite{DunEppGoo-JGAA-12}.

When a given graph is planar, we would like to find a planar Lombardi drawing of it. When this is possible, the resulting drawings have simple edge shapes, no crossings, and optimal angular resolution, all of which are properties that lead to more readable drawings. It is known that all Halin graphs have planar Lombardi drawings~\cite{DunEppGoo-JGAA-12}, that all 3-regular planar graphs~\cite{Epp-DCG-14} and all 4-regular polyhedral graphs~\cite{KinKobLof-GD-2017} have planar Lombardi drawings, and that all outerpaths have planar Lombardi drawings~\cite{DunEppGoo-JoCG-18}. For some other classes of planar graphs, even when a Lombardi drawing exists, it might not be planar.  Classes of planar graphs that are known to not always be drawable planarly in Lombardi style include the nested triangle graphs~\cite{DunEppGoo-JGAA-12}, 4-regular planar graphs~\cite{Epp-DCG-14}, planar 3-trees~\cite{DunEppGoo-JoCG-18}, and the graphs of knot and link diagrams~\cite{KinKobLof-GD-2017}.

However, for several other important classes of planar graphs, the existence of a planar Lombardi drawing has remained open. These include the outerplanar graphs, the series-parallel graphs, and the planar bipartite graphs. Outerplanar and series-parallel graphs are 2-degenerate, and always have Lombardi drawings. Planar bipartite graphs are 3-degenerate and such graphs usually have Lombardi drawings.\footnote{The only obstacle to Lombardi drawing for 3-degenerate graphs is the forced placement of two vertices on the same point, but the only examples for which this is known to happen are neither planar nor bipartite~\cite{DunEppGoo-JGAA-12}.}
However the known Lombardi drawings for these graphs are not necessarily planar. In this paper we settle this open problem for two of these classes of graphs, the planar bipartite graphs and the (embedded) series-parallel graphs. We construct a family of planar bipartite graphs whose Lombardi drawings are all nonplanar. We also construct a family of series-parallel graphs with a given embedding such that no planar Lombardi drawing respects that embedding.
 Our construction for series-parallel graphs can be extended to maximal series-parallel graphs, to bipartite series-parallel graphs and to apex-trees, the graphs formed by adding a single vertex to a tree.

\section{The graphs}

\begin{figure}[t]
\centering\includegraphics[width=0.8\columnwidth]{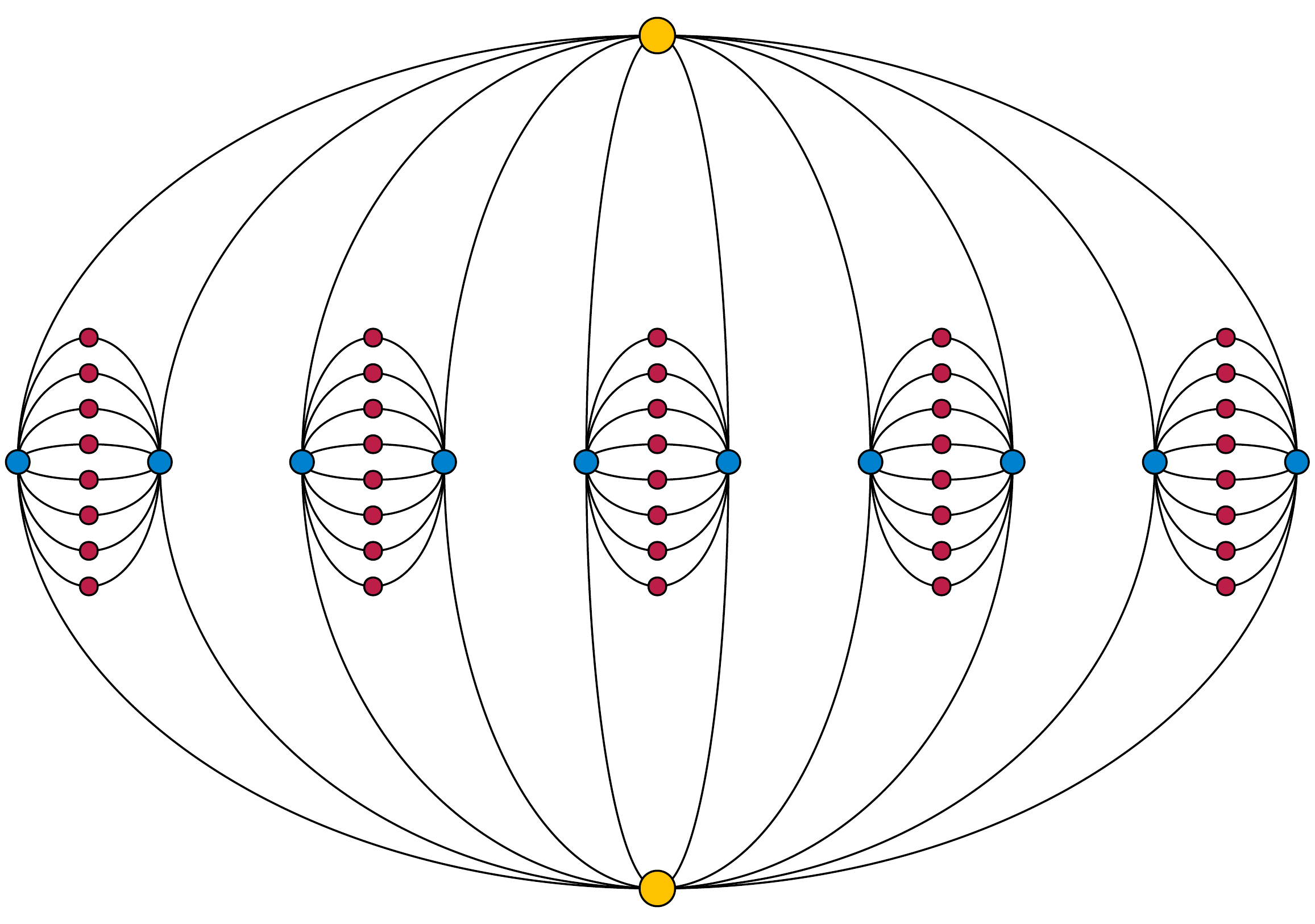}
\caption{The bipartite graph $B(5)$ formed by our construction. Although drawn planarly with curved edges, this is not a Lombardi drawing: the edges are arcs of ellipses rather than of circles, and pairs of consecutive edges at the same vertex do not all have the same angles.}
\label{fig:nested-K2n}
\end{figure}

We begin by describing the family of planar bipartite graphs $B(k)$ that we will prove (for sufficiently large $k$) do not have a planar Lombardi drawing. Each vertex in $B(k)$ has degree either $2$ or $2k$. To construct $B(k)$, begin with a complete bipartite graph $K_{2,2k}$ and its unique planar embedding; in \autoref{fig:nested-K2n}, the two-vertex side of the bipartition of this graph is shown by the yellow vertices and the $2k$-vertex side is shown by the blue vertices. Each yellow vertex has exactly $2k$ blue neighbors.

Next, partition the blue vertices into $k$ pairs of vertices, each sharing a face. For each pair of blue vertices in this partition, add another complete bipartite graph $K_{2,2k-2}$ connecting these two blue vertices to $2k-2$ additional vertices (shown as red in the figure). After this addition, each blue vertex has exactly $2k$ neighbors, two of them yellow and the rest red. Each red vertex has exactly two neighbors. There are two yellow vertices, $2k$ blue vertices, and $k(2k-2)$ red vertices,
for a total of $2k^2+2$ vertices in the overall graph.

Clearly, the graphs $B(k)$ are all planar, because they are formed by attaching together planar subgraphs (complete bipartite graphs where one side has two vertices) on pairs of vertices that are cofacial in both subgraphs. They are bipartite, with the yellow and red vertices on one side of their bipartition and the blue vertices on the other side. Although they are not $3$-vertex-connected, all of their planar embeddings are isomorphic.

\begin{figure}[t]
\centering\includegraphics[width=0.8\columnwidth]{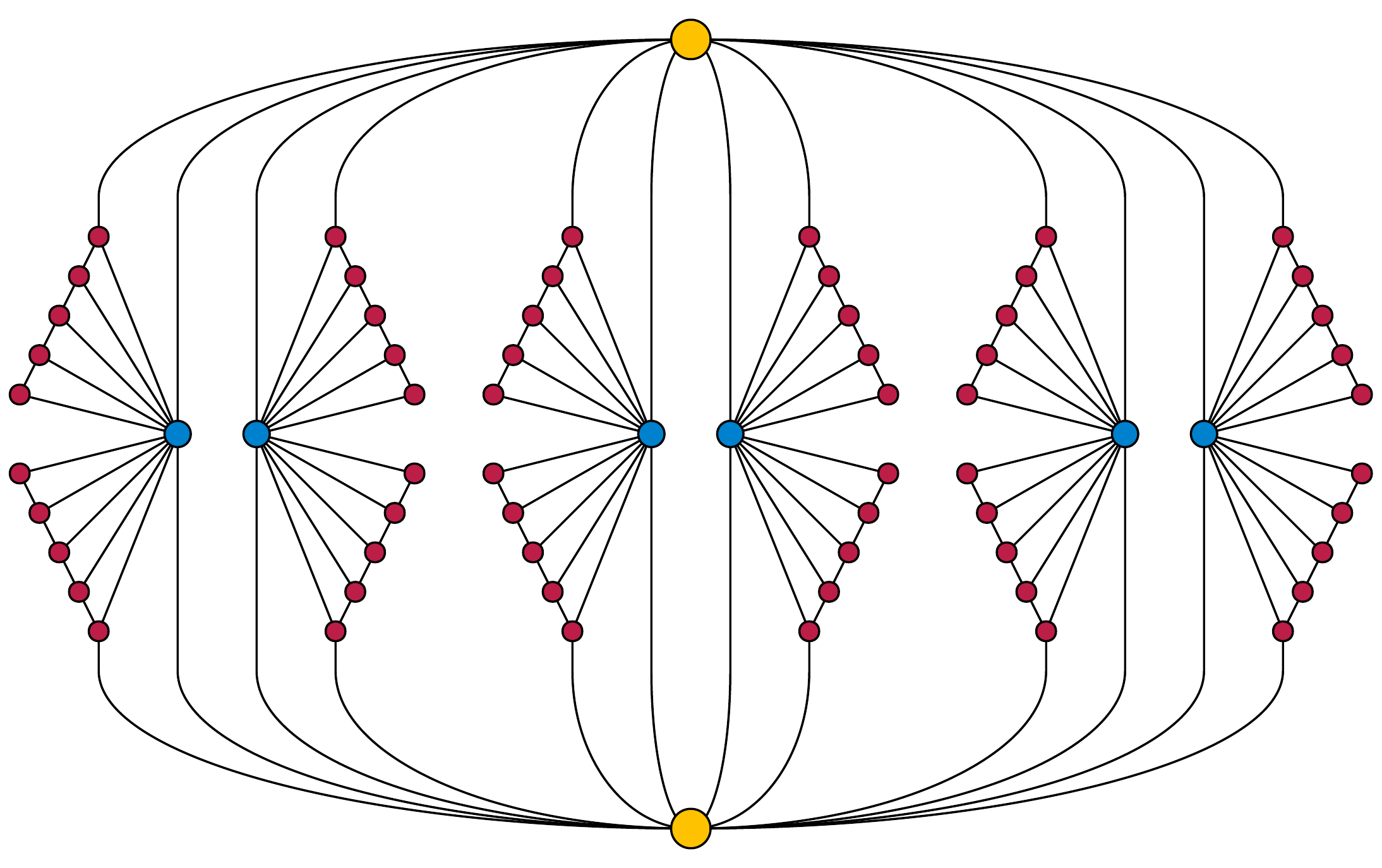}
\caption{The embedded series-parallel graph $S(3)$ formed by our construction.}
\label{fig:series-parallel}
\end{figure}

Analogously, we define a family of embedded series-parallel graphs $S(k)$. Again, each such graph will have two yellow vertices and $2k$ blue vertices, connected in the pattern of a complete bipartite graph $K_{2,2k}$. For each yellow--blue edge $e$ of this graph, we add a path of $2k-1$ red vertices.
We connect every vertex in this path to the blue endpoint of $e$, and we connect one endpoint of the path to the yellow endpoint of $e$. We fix an embedding of $S(k)$ in which every yellow--blue quadrilateral contains either zero or four red paths (\autoref{fig:series-parallel}). The resulting graph has two yellow vertices, $2k$ blue vertices, and $4k(2k-1)$ red vertices, for a total of $8k^2-2k+2$ vertices. The yellow and blue vertices have degree $4k$, while the red vertices have degrees two or three.

We claim that, for sufficiently large values of $k$, the graphs $G(k)$ and $S(k)$ do not have planar Lombardi drawings. Therefore, neither every planar bipartite graph nor every series-parallel graph has a planar Lombardi drawing. In the remainder of this paper we prove this claim.

\section{Equiangular arc-quadrilaterals}

The key feature of both of our graph constructions $B(k)$ and $S(k)$ is the existence of many yellow--blue quadrilateral faces in which all vertices have equal and high degree (this degree is $d=2k$ in $B(k)$ and $d=4k$ in $S(k)$). If such a graph is to have a Lombardi drawing, each of these faces must necessarily be drawn as a quadrilateral with circular-arc sides and with the same interior angle $2\pi/d$ at all four of its vertices. Equiangular arc-quadrilaterals have been investigated before from the point of view of conformal mapping~\cite{BroPor-CMFT-12}; in this section we investigate some of their additional properties.

Our main tool is the following lemma:

\begin{lemma}
\label{lem:cyclic}
Let $abcd$ be a non-self-crossing quadrilateral in the plane with circular-arc sides and equal interior angles.
Then the four points $abcd$ lie on a circle and the quadrilateral $abcd$ either lies entirely inside or entirely outside the circle.
\end{lemma}

\begin{proof}
We abbreviate the conclusion of the lemma by saying that $abcd$ is cyclic.
The properties of being an equiangular non-self-crossing circular-arc quadrilateral and of being cyclic are both invariant under M\"obius transformations, which preserve both circularity and the crossing angles of curves.
Therefore, if we can find a M\"obius transformation of a given equiangular circular-arc quadrilateral such that the transformed quadrilateral is cyclic, the original quadrilateral will also be cyclic, as the lemma states it to be.

Start by finding a M\"obius transformations which makes two opposite arcs $ab$ and $cd$ come from circles with the same radius as each other. Because of the equality of crossing angles, and by symmetry, both of the other two circular arcs $bc$ and $ad$ must come from circles whose centers lie on the perpendicular bisector of $ab$ and $cd$.
There remains a one--dimensional family of M\"obius transformations that preserve the position of the circles containing the transformed copies of arcs $ab$ and $cd$ but that move the other two circles along the bisector of these two fixed circles. We can use this remaining degree of freedom to move the other two circles so that their centers are equidistant from the midpoint of the centers of the two fixed circles.

\begin{figure}[t]
\centering\includegraphics[width=0.8\columnwidth]{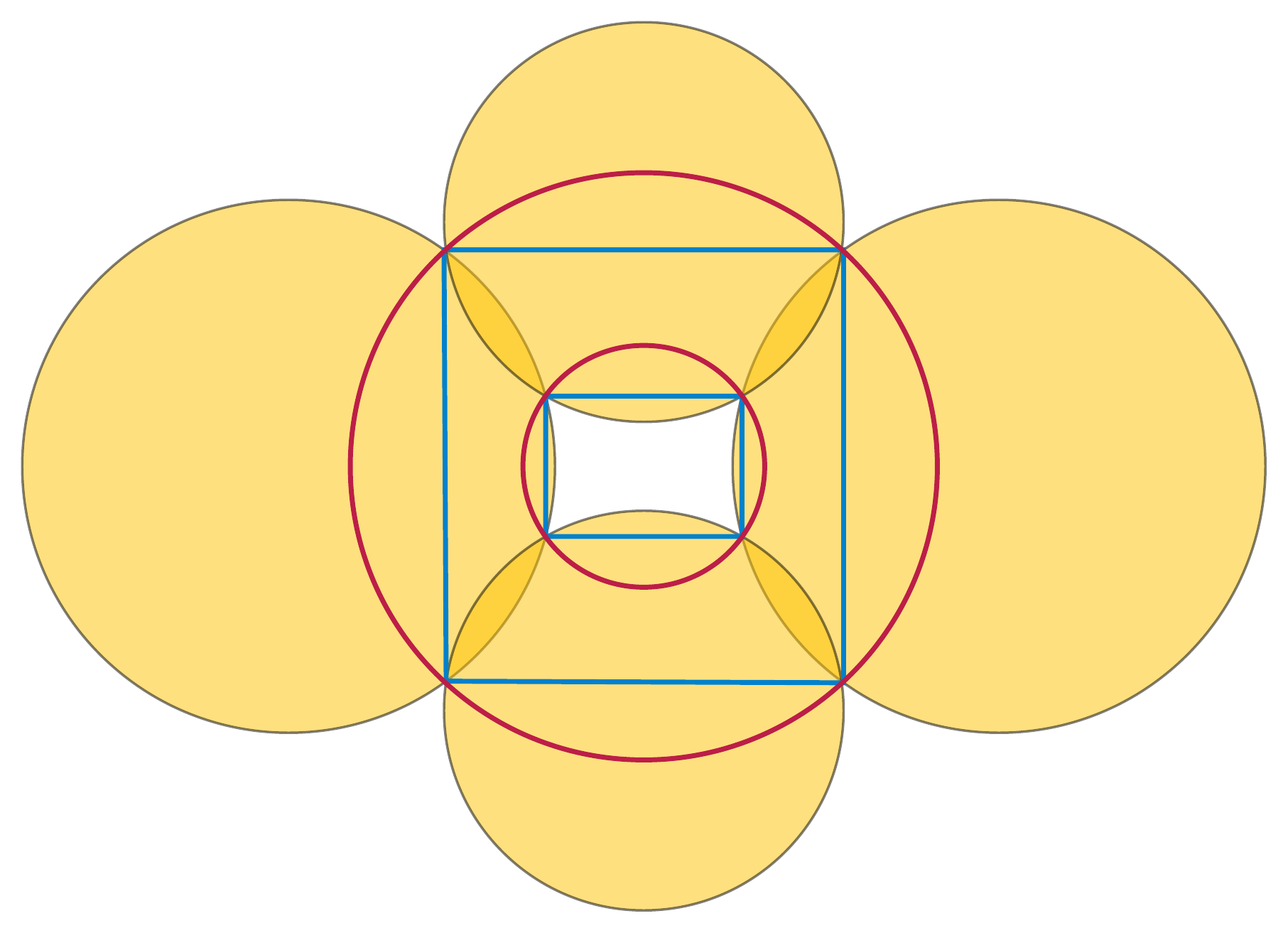}
\caption{Illustration for \autoref{lem:cyclic}: Four circles with centers on a rhombus, and with opposite pairs of circles having equal radii, define two rectangles of pairwise intersection points.}
\label{fig:4circle-rect}
\end{figure}

After this transformation, it follows from the equality of crossing angles that the circles containing the transformed copies of arcs $bc$ and $ad$ have the same radii as each other, and the four circles have been transformed into a position centered at the vertices of a rhombus with opposite pairs having the same radius as each other. By symmetry, the transformed copies of vertices $abcd$ must lie on one of the two rectangles defined by the crossing points of these four transformed circles.
(\autoref{fig:4circle-rect}).

If the interior angle of $abcd$ is less than $\pi$, then the transformed copy of $abcd$ must lie within the circle that circumscribes the inner rectangle, forming the boundary of a hole in the union of the four transformed disks. If the interior angle is greater than $\pi$, it must lie  outside the outer circle, forming the outer boundary of the union of the four disks. In either case, the transformed copy of $abcd$ is cyclic, so $abcd$ itself must be cyclic.
\end{proof}

A special case of this lemma for right-angled arc-quadrilaterals was used previously by the author to prove that some 4-regular planar graphs have no planar Lombardi drawing~\cite{Epp-DCG-14}.
Another special case, for arc-quadrilaterals in which all interior angles are zero, has been used previously in mesh generation~\cite{BerMitRup-DCG-95}.

\begin{definition}
We define the \emph{tilt} of an equiangular circular-arc quadrilateral to be the maximum interior angle of any of the four circular-arc bigons between the quadrilateral and its enclosing circle.
\end{definition}

Each of the four bigons has equal angles at its two vertices. At each of the four vertices, the two bigon angles and the interior angle of the quadrilateral add to $\pi$. It follows that opposite bigons have the same angles as each other, and each vertex of the quadrilateral is incident to a bigon with vertices of the tilt angle.

\section{Bipolar coordinates}

\begin{figure}[t]
\centering\includegraphics[width=0.65\columnwidth]{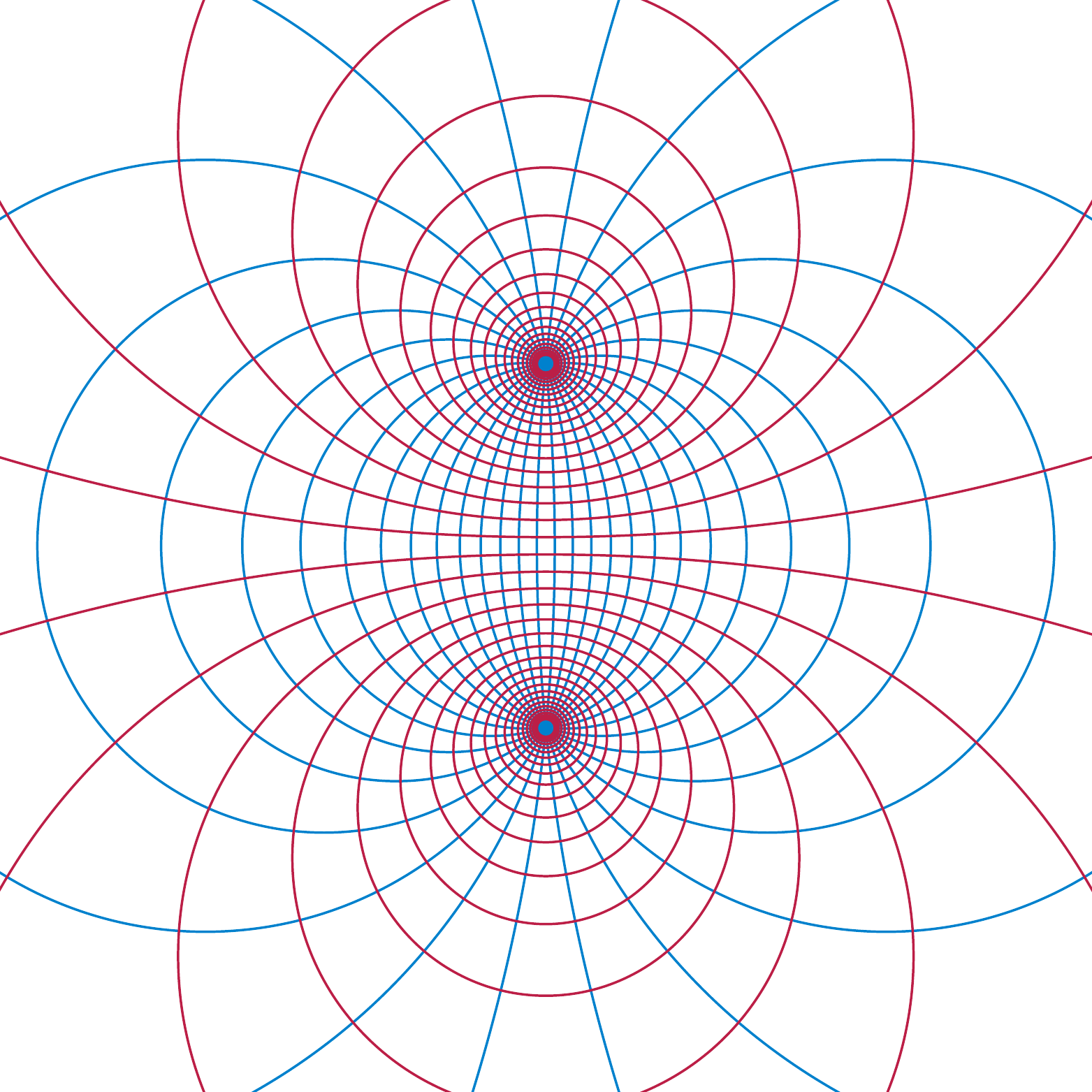}
\caption{Curves of constant and evenly-spaced coordinate values for bipolar coordinates, forming two orthogonal pencils of circles.}
\label{fig:apollo}
\end{figure}

To describe a second parameter of equiangular circular-arc quadrilaterals, it is convenient to introduce the \emph{bipolar coordinate system}, defined from a pair of points $s$ and $t$, the \emph{foci} of the coordinate system.
These coordinates are conventionally denoted $\sigma$ and $\tau$.
The $\sigma$-coordinate $\sigma_p$ of a point $p$ is the (oriented) angle $spt$, whose level sets are the blue circular arcs through the two foci in \autoref{fig:apollo}. The $\tau$-coordinate $\tau_p$ of $p$ is the logarithm of the ratio of the two distances from $p$ to the two foci, whose level sets are the red circles separating the two foci in \autoref{fig:apollo}.
This coordinate system has the convenient property that any (orientation-preserving) M\"obius transformation that preserves the location of the two foci acts by translation on the coordinates.\footnote{This property can be seen as a reflection of the fact that the bipolar coordinate system comes from a conformal mapping of a rectangular grid; see, e.g., \cite{CheTsaLiu-CAEE-09}.}

\begin{lemma}
Any M\"obius transformation that preserves the location of the two foci acts on the bipolar coordinates of any point by adding a fixed value to its $\sigma$-coordinate (modulo $2\pi$) and adding another fixed value to its $\tau$-coordinate, with the added values depending on the transformation but not on the point.
\end{lemma}

\begin{proof}
All  M\"obius transformations preserves circles, incidences between points and curves, and angles between pairs of incident curves. Therefore, any focus-preserving M\"obius transformation takes circles through the two foci (the level sets for $\sigma$-coordinates) to other circles through the two foci, and it takes the perpendicular family of circles (the level sets for $\tau$-coordinates) to other circles in the same family. Therefore it acts separately on the $\sigma$- and $\tau$-coordinates. The additivity of its action on the $\sigma$-coordinates follows from the preservation of angles between pairs of circles through the two foci.

To show that the transformation acts additively on $\tau$-coordinate (the logarithm of the  ratio of distances of a point from the two foci), we can assume without loss of generality (by scaling, translating, and rotating the plane if necessary) that the two foci are at the two points $q=\pm 1$ of the complex plane. Consider the general form
$q\mapsto (aq+b)/(cq+d)$ of a M\"obius transformation as a fractional linear transformation of the complex plane. For a transformation to fix $q=1$ we need $a+b=c+d$ and for it to fix $q=-1$ we need $b-a=-(d-c)$. Solving these two equations in four unknowns gives $a=d$ and $b=c$. Therefore, the transformations fixing the foci take the special form $q\mapsto (aq+b)/(bq+a)$.

For a transformation of this form, and for any point $x$ on the interval $[-1,1]$ of real numbers with distance ratio $(1+x)/(1-x)$,
the image of $x$ has distance ratio
\[\frac{1+(ax+b)(bx+a)}{1-(ax+b)(bx+a)}
=\frac{a+b}{a-b}\cdot\frac{1+x}{1-x},\]
multiplying the original distance ratio of $x$ by a value that depends only on the transformation. Because the transformation acts separately on $\sigma$- and $\tau$-coordinates, we obtain the same multiplicative action on distance ratios for any other point on the complex plane with the same $\tau$-coordinate as $x$.
This multiplicative action on distance ratios translates into an additive action on their logarithms.
\end{proof}

\begin{figure}[t]
\centering\includegraphics[width=0.5\columnwidth]{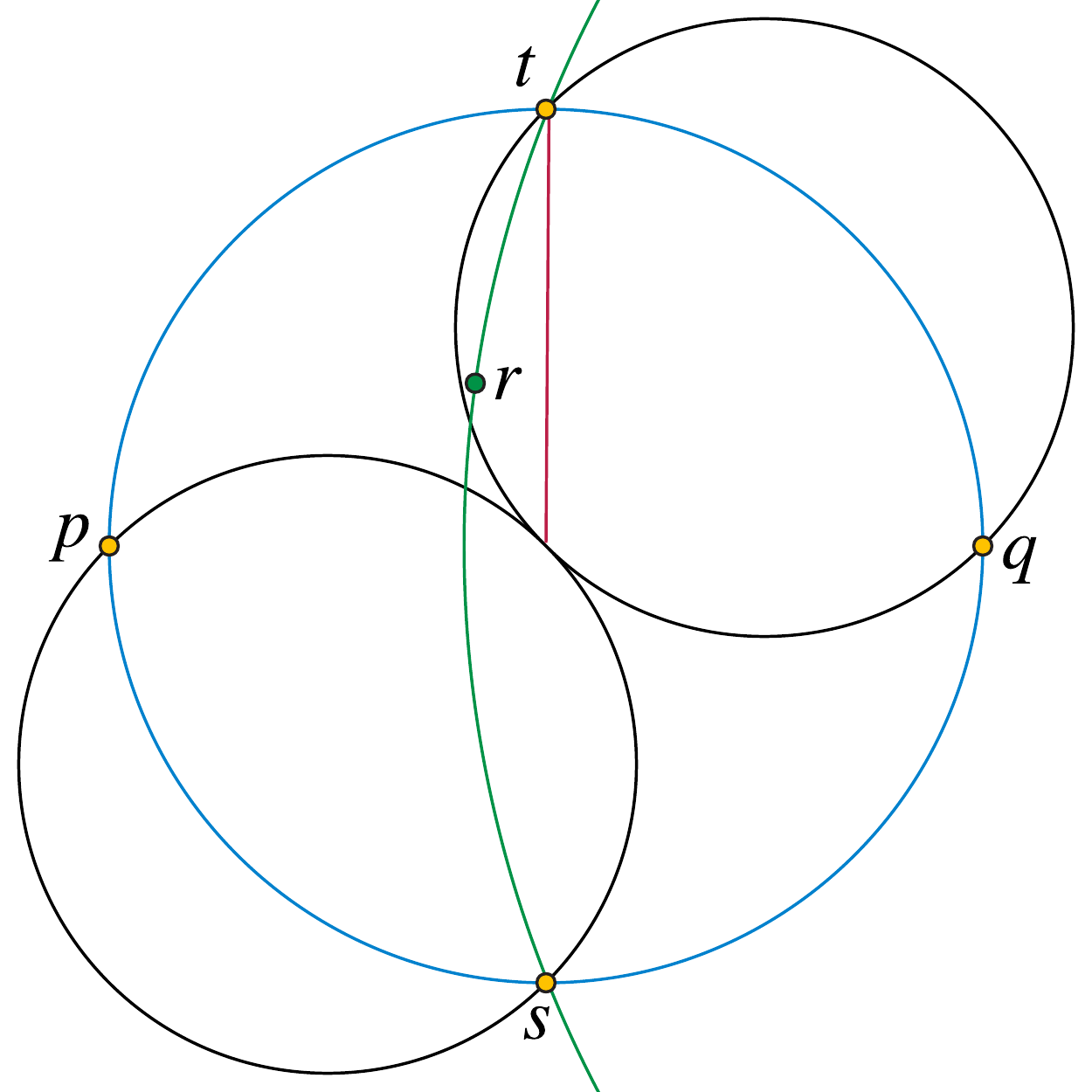}
\caption{Illustration for \autoref{lem:lift}: If quadrilateral $sptq$ has high tilt, and $r$ lies between the quadrilateral and its enclosing circle, to the left of the red arc, then $r$ must have a higher $\tau$-coordinate value than $p$.}
\label{fig:lift}
\end{figure}

Another advantage of bipolar coordinates is that they provide a way of comparing angles at the two foci, that will be convenient for relating the tilts of different quadrilaterals to each other:

\begin{observation}
\label{obs:tilt-angle}
Let $sptq$ be an equiangular arc-quadrilateral with interior angle $\theta$ and tilt $\varphi$.
Then, in the bipolar coordinate system with foci $s$ and $t$,
the angle (difference between the $\sigma$-coordinates) of arc $tp$ in the limit as it approaches $t$ and of  arc $sq$ as it approaches $s$ 
is exactly $2\theta+2\varphi-\pi$.
\end{observation}

We can also use bipolar coordinates to show that heavily tilted quadrilaterals lead to an increase in $\tau$-coordinate:

\begin{lemma}
\label{lem:lift}
Let $sptq$ be an equilateral arc-quadrilateral with tilt at least $3\pi/4$, such that the large angle between vertex $s$ and the circle $C$ containing the quadrilateral is on the clockwise side of $s$ (the side closest to $p$).
Let $r$ be a point in the bigon between arc $tq$ and $C$, such that circular arc $srt$ makes an angle of at most $\pi/2$ with circular arc $spt$. Then, in the bipolar coordinate system for foci $s$ and $t$, $\tau_r > \tau_p$.
\end{lemma}

\begin{proof}
Because of the M\"obius invariance of coordinate differences in the bipolar coordinate system, we can without loss of generality perform a M\"obius transformation so that $s$, $p$, and $t$ are the bottom, left, and topmost points of $C$, as shown in \autoref{fig:lift}. After this transformation, points above the horizontal line through $p$ will have higher $\tau$-coordinate than $o$, and points below the horizontal line through $o$ will have lower $\tau$-coordinate.

As the figure shows, an arc with tilt exactly $3\pi/4$ through $p$ and $s$ passes through the center of circle $C$, causing the region in which $r$ may lie to be bounded by a vertical line segment (red) from the circle's center to $t$. All points within this region have higher $\tau$-coordinate than $p$. For tilt values greater than $3\pi/4$, the arc from $p$ to $s$ with that tilt extends even farther beyond the center of $C$, so (although arc $tq$ may also extend farther to the left) the region in which $r$ may lie remains bounded within the upper left quarter of $C$, within which all $\tau$-coordinates are greater than that of $p$.
\end{proof}

\section{Nonplanarity}

We are now ready for our main theorems.

\begin{theorem}
\label{thm:bipartite-nonplanar}
For $k>8$, the bipartite graph $B(k)$ does not have a planar Lombardi drawing.
\end{theorem}

\begin{proof}
Let $s$ and $t$ be the two yellow vertices of the graph $B(k)$.
We will consider a bipolar coordinate system with foci $s$ and $t$.
Note that graph $B(k)$ contains $k$ quadrilateral faces $sp_itq_i$, where $p_i$ and $q_i$ are blue vertices.
Because all four vertices of these quadrilaterals have degree $2k$, these quadrilaterals must be drawn (if a planar Lombardi drawing is to exist) as equiangular arc-quadrilaterals with interior angle $\pi/k$.

\begin{figure}[t]
\centering\includegraphics[width=0.8\columnwidth]{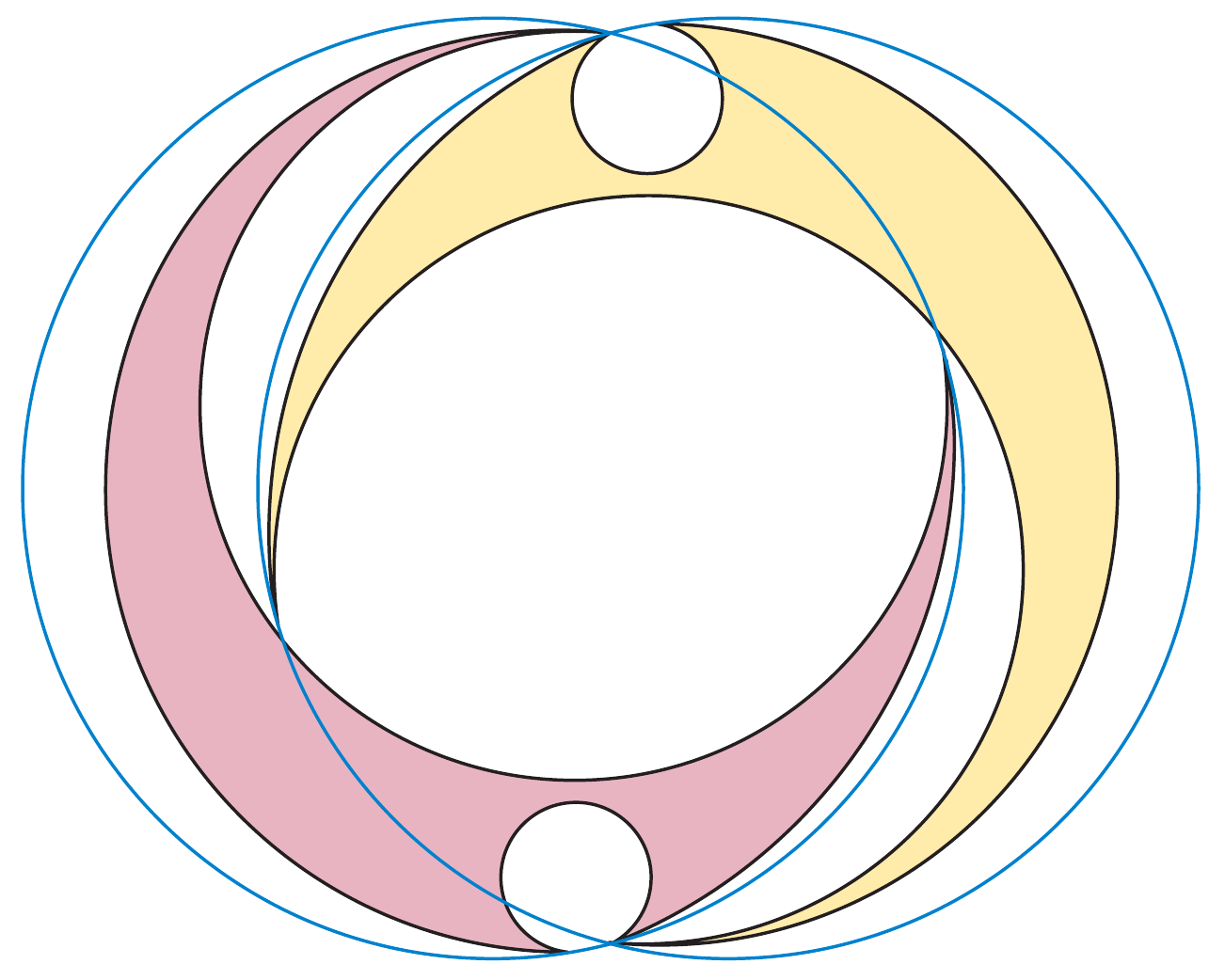}
\caption{Two arc-quadrilaterals with sharp angles between them at their shared vertices reach into each other's pockets to touch their circumscribing circles. The two smaller pockets on the outer arcs of the circles have significantly different $\tau$-coordinates from each other in bipolar coordinates with the shared vertices as foci.}
\label{fig:reacharound}
\end{figure}

Consider the two consecutive quadrilaterals $sp_itq_i$ and $sp_{i+1}tq_{i+1}$ whose enclosing circles $C_i$ and $C_{i+1}$ meet each other at the sharpest angle of any two consecutive enclosing circles. The sum of the angles between the $k$ consecutive circles is $2\pi$ so this minimum angle is at most $2\pi/k$. In order for point $q_i$ to lie on circle $C_i$, some arc of circle $C_i$ must lie on the same side as $q_i$ of quadrilateral $sp_{i+1}tq_{i+1}$. This arc must stay outside of quadrilateral $sp_{i+1}tq_{i+1}$ from its crossing point with the quadrilateral until terminating at either $s$ or $t$; by symmetry, we can assume without loss of generality that it terminates at the lower vertex $s$, as shown in \autoref{fig:reacharound}. Then, near $s$, quadrilateral $sp_{i+1}tq_{i+1}$ lies between circles $C_i$ and $C_{i+1}$, so it must have tilt at least $\pi(1-\tfrac{2}{k})$. By \autoref{obs:tilt-angle} and the equal spacing of angles around $s$ and $t$, all quadrilaterals $sp_itq_i$ must have the same tilt.

Because $k>8$, this tilt is $\ge 3\pi/4$, so each quadrilateral $sp_itp_i$ meets the precondition of having high tilt of \autoref{lem:lift}. For any quadrilateral $sp_itq_i$, the point $p_{i+1}$ of the next quadrilateral is connected to $s$ by an arc of quadrilateral $sp_{i+1}tq_{i+1}$ that lies entirely within $C_{i+1}$ and makes an angle of $\pi/k$ to arc $sq_i$, 
so the arc of $C_{i+1}$ containing $p_{i+1}$ makes an angle of at most $3\pi/k$ to the arc of $C_i$ containing $p_i$.
Thus, point $p_{i+1}$ meets the other precondition of \autoref{lem:lift} for the position of the point $r$ with respect to the quadrilateral.
By this lemma, each point $p_{i+1}$ has a greater $\tau$-coordinate than $p_i$. But it is impossible for this monotonic increase in $\tau$-coordinates to continue all the way around the circle of quadrilaterals surrounding the two foci and back to the starting point. This impossibility shows that the drawing cannot exist.
\end{proof}

\begin{theorem}
For $k>8$ the series-parallel graph $S(k)$, embedded as shown in \autoref{fig:series-parallel}, does not have a planar Lombardi drawing.
\end{theorem}

\begin{proof}
As with $B(k)$, this graph contains $k$ quadrilateral faces, sharing the same two opposite yellow vertices, in which all vertices have equal degree ($4k$ in $S(k)$ instead of $2k$ in $B(k)$). The proof of \autoref{thm:bipartite-nonplanar} used only this property of $B(k)$, and not the precise value of the interior angle of these quadrilaterals, so it applies equally well to $S(k)$.
\end{proof}

We remark that the construction of $S(k)$ can be adjusted in several different ways to obtain more constrained families of embedded series-parallel and related graphs that, again, have no planar Lombardi drawing:
\begin{itemize}
\item If we add an edge between the two yellow vertices, and adjust the lengths of the red chains to keep the yellow and blue degrees equal, we obtain a family of embedded maximal series-parallel graphs (that is, embedded 2-trees) with no planar Lombardi drawing.
\item If we subdivide the yellow--red and red--red edges of $S(k)$, we obtain a family of embedded bipartite series-parallel graphs with no planar Lombardi drawing.
\item If we replace the red chains of $S(k)$ by an appropriate number of degree-one red vertices, connected to the blue vertices, we obtain a family of embedded apex-trees (graphs formed by adding a single vertex to a tree) with no planar Lombardi drawing. The apex vertex (the vertex whose removal produces a tree) can be chosen to be either of the two yellow vertices.
\end{itemize}
We omit the details.

\section{Conclusions}

We have shown that bipartite planar graphs, and series-parallel graphs with a fixed planar embedding, do not always have planar Lombardi drawings, even though their low degeneracy implies that they always have (nonplanar) Lombardi drawings. In the question of which important subfamilies of planar graphs have planar Lombardi drawings, several important cases remain unsolved. These include the outerplanar graphs, both with and without assuming an outerplanar embedding, the cactus graphs, and the series-parallel graphs without a fixed choice of embedding. We leave these as open for future research.

\bibliographystyle{amsplainurl}
\bibliography{lombardi}

\providecommand{\bysame}{\leavevmode\hbox to3em{\hrulefill}\thinspace}
\providecommand{\MR}{\relax\ifhmode\unskip\space\fi MR }
\providecommand{\MRhref}[2]{%
  \href{http://www.ams.org/mathscinet-getitem?mr=#1}{#2}%
}
\providecommand{\href}[2]{#2}
\begin{thebibliography}{1}

\bibitem{BerMitRup-DCG-95}
Marshall Bern, Scott Mitchell, and Jim Ruppert, \emph{{Linear-size nonobtuse
  triangulation of polygons}}, Discrete {\&} Computational Geometry \textbf{14}
  (1995), no.~4, 411{--}428, \href {http://dx.doi.org/10.1007/BF02570715}
  {\path{doi:10.1007/BF02570715}}, \MR{1360945}.

\bibitem{BroPor-CMFT-12}
Philip~R. Brown and R.~Michael Porter, \emph{{Conformal mapping of circular
  quadrilaterals and Weierstrass elliptic functions}}, Computational Methods
  and Function Theory \textbf{11} (2011), no.~2, 463{--}486, \href
  {http://dx.doi.org/10.1007/BF03321872} {\path{doi:10.1007/BF03321872}},
  \MR{2858958}.

\bibitem{CheTsaLiu-CAEE-09}
Jeng-Tzong Chen, Ming-Hong Tsai, and Chein-Shan Liu, \emph{{Conformal mapping
  and bipolar coordinate for eccentric Laplace problems}}, Computer
  Applications in Engineering Education \textbf{17} (2009), no.~3, 314{--}322,
  \href {http://dx.doi.org/10.1002/cae.20208} {\path{doi:10.1002/cae.20208}}.

\bibitem{DunEppGoo-JoCG-18}
Christian~A. Duncan, David Eppstein, Michael~T. Goodrich, Stephen~G. Kobourov,
  Maarten L{\"o}ffler, and Martin N{\"o}llenburg, \emph{{Planar and poly-arc
  Lombardi drawings}}, Journal of Computational Geometry \textbf{9} (2018),
  no.~1, 328{--}355, \href {http://arxiv.org/abs/1109.0345}
  {\path{arXiv:1109.0345}}, \href {http://dx.doi.org/10.20382/jocg.v9i1a11}
  {\path{doi:10.20382/jocg.v9i1a11}}, \MR{3855883}.

\bibitem{DunEppGoo-JGAA-12}
Christian~A. Duncan, David Eppstein, Michael~T. Goodrich, Stephen~G. Kobourov,
  and Martin N{\"o}llenburg, \emph{{Lombardi drawings of graphs}}, J. Graph
  Algorithms {\&} Applications \textbf{16} (2012), no.~1, 85{--}108, \href
  {http://arxiv.org/abs/1009.0579} {\path{arXiv:1009.0579}}, \href
  {http://dx.doi.org/10.7155/jgaa.00251} {\path{doi:10.7155/jgaa.00251}},
  \MR{2872431}.

\bibitem{DunEppGoo-DCG-13}
\bysame, \emph{{Drawing trees with perfect angular resolution and polynomial
  area}}, Discrete {\&} Computational Geometry \textbf{49} (2013), no.~2,
  157{--}182, \href {http://arxiv.org/abs/1009.0581} {\path{arXiv:1009.0581}},
  \href {http://dx.doi.org/10.1007/s00454-012-9472-y}
  {\path{doi:10.1007/s00454-012-9472-y}}, \MR{3017904}.

\bibitem{Epp-DCG-14}
David Eppstein, \emph{{A M{\"o}bius-invariant power diagram and its
  applications to soap bubbles and planar Lombardi drawing}}, Discrete {\&}
  Computational Geometry \textbf{52} (2014), no.~3, 515{--}550, \href
  {http://dx.doi.org/10.1007/s00454-014-9627-0}
  {\path{doi:10.1007/s00454-014-9627-0}}, \MR{3257673}.

\bibitem{KinKobLof-GD-2017}
Philipp Kindermann, Stephen~G. Kobourov, Maarten L{\"o}ffler, Martin
  N{\"o}llenburg, Andr{\'e} Schulz, and Birgit Vogtenhuber, \emph{{Lombardi
  drawings of knots and links}}, Graph Drawing and Network Visualization - 25th
  International Symposium, GD 2017, Boston, MA, USA, September 25-27, 2017,
  Revised Selected Papers (Fabrizio Frati and Kwan{-}Liu Ma, eds.), Lecture
  Notes in Computer Science, vol. 10692, Springer, 2017, pp.~113{--}126, \href
  {http://arxiv.org/abs/1708.09819} {\path{arXiv:1708.09819}}, \href
  {http://dx.doi.org/10.1007/978-3-319-73915-1_10}
  {\path{doi:10.1007/978-3-319-73915-1_10}}.

\end{thebibliography}

\end{document}